\documentclass{article}

\usepackage{fancyhdr}

\usepackage{amssymb}
\usepackage{graphicx}
\usepackage{bm}
\usepackage{csquotes}
\usepackage{enumitem}
\usepackage{hyperref}
\usepackage{listings}
\usepackage{mathtools}
\usepackage{amsmath}
\usepackage{amsthm}
\usepackage{nicefrac}
\usepackage{todonotes}
\usepackage{xargs}
\usepackage{xfrac}
\usepackage{xspace}
\usepackage{dsfont}
\usepackage{relsize}

\usepackage{url}
\urldef{\mailsa}\path|{alexander.maecker, manuel.malatyali, fmadh, soeren.riechers}@upb.de|

\usepackage{mathtools}
\newcommand*{\OPT}{\ensuremath{\textsc{Opt}}\xspace}
\newcommand*{\PROBLEM}{\textsc{CloudScheduling}\xspace}
\newcommand*{\ALGO}{\textsc{BatchedDispatch}\xspace}
\newcommand{\setup}[1]{s_{\!\mathsmaller{#1}}}
\newcommand{\size}[2]{p_{#1,{\!\!\mathsmaller{~#2}}}}
\newcommand{\slack}[2]{\sigma_{#1,{\!\!\mathsmaller{~#2}}}}

\newtheorem{theorem}{Theorem}[section]
\newtheorem{lemma}[theorem]{Lemma}

\newtheorem{proposition}[theorem]{Proposition}

\begin{document}

\title{Cost-efficient Scheduling\\ on Machines from the Cloud\thanks{This work was partially supported by the German Research Foundation (DFG)
within the Collaborative Research Centre ``On-The-Fly Computing'' (SFB 901)}~~\footnote{A conference version of this paper is published in the proceedings of the 10th Annual International Conference on Combinatorial Optimization and Applications (COCOA). The final publication is available at Springer via \url{http://dx.doi.org/10.1007/978-3-319-48749-6_42}.}}
\author{Alexander M\"acker \and Manuel Malatyali \and Friedhelm Meyer auf der Heide \and S\"oren Riechers \\ [0.4em]
	Heinz Nixdorf Institute \& 	Computer Science Department\\
	Paderborn University, Germany\\[0.2em]
	\{amaecker, malatya, fmadh, soerenri\}@hni.upb.de}
	\date{}
\maketitle

\begin{abstract}
We consider a scheduling problem where machines need to be rented from the cloud in order to process jobs.
There are two types of machines available which can be rented for machine-type dependent prices and for arbitrary durations.
However, a machine-type dependent setup time is required before a machine is available for processing.
Jobs arrive online over time, have machine-type dependent sizes and have individual deadlines.
The objective is to rent machines and schedule jobs so as to meet all deadlines while minimizing the rental cost.

Since we observe the slack of jobs to have a fundamental influence on the competitiveness, we study the model when instances are parameterized by their (minimum) slack.
An instance is called to have a slack of $\beta$ if, for all jobs, the difference between the job's release time and the latest point in time at which it needs to be started is at least $\beta$.
While for $\beta < s$ no finite competitiveness is possible, our main result is an $O(\nicefrac{c}{\varepsilon} + \nicefrac{1}{\varepsilon^3})$-competitive online algorithm for $\beta = (1+\varepsilon)s$ with $\nicefrac{1}{s} \leq \varepsilon \leq 1$, where $s$ and $c$ denotes the largest setup time and the cost ratio of the machine-types, respectively.
It is complemented by a lower bound of $\Omega(\nicefrac{c}{\varepsilon})$.
\end{abstract}

\section{Introduction}
Cloud computing provides a new concept for provisioning computing power, usually in the form of virtual machines (VMs), that offers users potential benefits but also poses new (algorithmic) challenges to be addressed.
On the positive side, main advantages for users of moving their business to the cloud are manifold:
Possessing huge computing infrastructures as well as asset and maintenance cost associated with it are no longer required and hence, costs and risks of large investments are significantly reduced.
Instead, users are only charged to the extent they actually use computing power 
and it can be scaled up and down depending on current demands.
In practice, two ways of renting machines from the cloud are typically available -- on-demand-instances without long-term commitment, and reserved instances.  
In the former case, machines can be rented for any duration and the charging is by the hour (Amazon EC2 \cite{ec2}) or by the minute (Google Cloud \cite{google}); whereas in the latter case, specific (long-term) leases of various lengths are available. 
We focus on the former case and note that the latter might better be captured within the leasing framework introduced in \cite{meyerson05}.

Despite these potential benefits, the cloud also bears new challenges in terms of renting and utilizing resources in a (cost-)efficient way. 
Typically, the cloud offers diverse VM types differing in the provided resources and further characteristics.
For example, one might think of high-CPU instances, which especially suit the requirements for compute-intensive jobs, and high-memory instances for I/O-intensive jobs.
The performance of a job might then strongly depend on the VM type on which it is executed.
Hence, deliberate decisions on the machines to be rented need to be taken in order to be cost-efficient while guaranteeing a good performance, for example, in terms of user-defined due dates or deadlines implicitly given by a desired quality of service level.
Also, despite the fact that computing power can be scaled according to current demands, one needs to take into account that this scaling is not instantaneous and it initially could take some time for a VM to be ready for processing workload. 
A recent study \cite{mao12} shows such setup times to typically be in the range of several minutes for common cloud providers and hence, to be non-ignorable for the overall performance~\cite{mao10}.

\subsection{Problem Description}
We study the problem \PROBLEM, which is extracted from the preceding observations.
It studies a basic variant exploiting the heterogeneity of available machines and workloads by introducing the option for a scheduling algorithm to choose between machines of two types.
This approach can, for instance, be motivated by the fact that certain computations may be accelerated by the use of GPUs.
As one example, in \cite{cpugpu} it is observed that, for certain workloads, one can significantly improve performance when not only using classical CPU but additionally GPU instances in Amazon EC2.
\paragraph{Machine Environment}
There are two types $\tau \in \{A,B\}$ of machines, which can be rented in order to process workload: 
A machine $M$ of type $\tau$ can be opened at any time $a_M$ and, after a setup taking some non-negligible $s_\tau$ time units was performed, can be closed at any time $b_M$ with $b_M \geq a_M+s_\tau$, for values $a_M, b_M, s_\tau \in \mathbb{R}_{\geq 0}$.
It can process workload during the interval $[a_M+s_\tau, b_M]$ and
the rental cost incurred by a machine $M$ is determined by the duration for which it is open and is formally given by $c_\tau \cdot (b_M - a_M)$.
Note that a machine cannot process any workload nor can be closed during the setup.
However, the user is still charged for the duration of the setup.
We refer to the cost $c_\tau \cdot s_\tau$ incurred by a machine during its setup as \emph{setup cost} and to the cost $c_\tau \cdot (b_M-a_M-s_\tau)$ incurred during the remaining time it is open as \emph{processing cost}.
Without loss of generality, we assume that $\setup{B} \geq \setup{A}$, $c_A =1$ and define $c \coloneqq c_B$.
We restrict ourselves to the case $c\geq 1$ (which is consistent with practical observations where $c$ is usually in the range of $1$ to a few hundred \cite{mao12,ec2}).
However, one can obtain similar results for $c<1$ by an analogous reasoning.

\paragraph{Job Characteristics}
The workload of an instance is represented by a set $J$ of $n$ jobs, where each job $j$ is characterized by a release time $r_j \in \mathbb{R}_{\geq 0}$, a deadline $d_j \in \mathbb{R}_{>0}$, and sizes $p_{j,\tau} \in \mathbb{R}_{>0}$ for all $\tau \in \{A,B\}$, describing the processing time of job $j$ when assigned to a machine of type $\tau$.
Observe that therefore machines of the same type are considered to be identical while those of different types are unrelated.
Throughout the paper, we assume by using a suitable time scale that $p_{j,\tau} \geq 1$ for all $j \in J$ and all $\tau \in \{A,B\}$.
The jobs arrive online over time at their release times at which the scheduler gets to know the deadline and sizes of a job.
Each job needs to be completely processed by \emph{one} machine before its respective deadline, i.e., to any job $j$ we have to assign a \emph{processing interval} $I_j$ of length $p_{j,\tau}$ that is contained in $j$'s \emph{time window} $[r_j, d_j] \supseteq I_j$ on a machine $M$ of type $\tau$ that is open during the entire interval $I_j \subseteq [a_M+s_\tau, b_M]$.

A machine $M$ of type $\tau$ is called \emph{exclusive} machine for a job $j$ if it only processes $j$ 
and $b_M - a_M = s_\tau + \size{j}{\tau}$.
We define the \emph{slack} of a job $j$ as the amount by which $j$ is shiftable in its window: $\sigma_{j,\tau} \coloneqq d_j - r_j - p_{j,\tau}$ for $\tau \in \{A,B\}$.

\paragraph{Objective Function}
The objective of \PROBLEM is to rent machines and compute a feasible schedule that minimizes the rental cost, i.e., $\sum_M c_{\tau(M)} \cdot (b_M-a_M)$, where $\tau(M)$ denotes the type of machine $M$.

\paragraph{Quality Measure}
We analyze the quality of our algorithms in terms of their competitiveness and assume that problem instances are parameterized by their \emph{minimum slack}.
An instance is said to have a minimum slack of $\beta$ if $\max_\tau\{\slack{j}{\tau}\} \geq \beta$, for all $j \in J$.
Then, for a given $\beta$, an algorithm is called $\rho$-\emph{competitive} if, on any instance with minimum slack $\beta$, the rental cost is at most by a factor $\rho$ larger than the cost of an optimal offline algorithm.

In the following, by \OPT we denote an optimal schedule as well as the cost it incurs.
Throughout the paper (in all upper and lower bounds), we assume that \OPT is not too small and in particular, $\OPT =\Omega(c \cdot r_{max})$, where $r_{max} = \max_{j \in J} r_j$.
This bound is true, for example, if the optimal solution has to maintain at least one open machine (of the more expensive machine type) during (a constant fraction of) the considered time horizon.
This assumption seems to be reasonable, particularly for large-scale systems where, at any time, the decision to make is rather concerned with dozens of machines than whether a single machine is rented at all.
Similar assumptions are made in the literature \cite{azar13}, where it is argued (for identically priced machines) that the case where $\OPT = o(r_{max})$ is of minor interest as workload and costs are very small.

\subsection{Related Work}
\label{sec:relatedWork}
Cloud scheduling has recently attracted the interest of theoretical researchers.
Azar et al.~\cite{azar13} consider a scheduling problem where jobs arrive online over time and need to be processed by identical machines rented from the cloud.
While machines are paid in a pay-as-you-go manner, a fixed setup time $T_s$ is required before the respective machine is available for processing.
In this setting, Azar et al.\ consider a bicriterial optimization problem where the rental cost is to be minimized while guaranteeing a reasonable maximum delay.
An online algorithm that is $(1+ \varepsilon, O(1/\varepsilon))$-competitive regarding the cost and the maximum delay, respectively, is provided.
A different model for cloud scheduling was considered by Saha \cite{saha13}.
She considers jobs that arrive over time and need to be finished before their respective deadlines.
To process jobs, identical machines are available and need to be rented from the cloud.
The goal is to minimize the rental cost where a machine that is rented for $t$ time units incurs cost of $\lceil t /D \rceil$ for some fixed $D$. 
The problem is considered as an offline as well as an online problem and algorithms that guarantee solutions incurring costs of $O(\alpha)\OPT$ (where $\alpha$ is the approximation factor of the algorithm for machine minimization in use, see below) and $O(\log(p_{max}/p_{min}))\OPT$, respectively, are designed.

A different, but closely related problem is that of \emph{machine minimization}.
In this problem, $n$ jobs with release times and deadlines are considered and the objective is to finish each job before its deadline while minimizing the number of machines used. 
This problem has been studied in online and offline settings for the general and different special cases.
The first result is due to \cite{raghavan87} where an offline algorithm with approximation factor of $O(\log n/ \log \log n)$ is given.
This was later improved to $O(\sqrt{\log n/ \log \log n})$ by Chuzhoy et al.\ \cite{chuzhoy04}.
Better bounds have been achieved for special cases; if all jobs have a common release date or equal processing times, constant approximation factors are achieved \cite{yu09}.
In the online case, a lower bound of $\Omega(n + \log(p_{max}/p_{min}))$ and an algorithm matching this bound is given in \cite{saha13}.
For jobs of equal size, an optimal $e$-competitive algorithm is presented in \cite{devanur14}.

A further area of research that studies rental/leasing problems from an algorithmic perspective and which recently gained attention is that of \emph{resource leasing}.
Its focus is on infrastructure problems and while classically these problems are modeled such that resources are bought and then available all the time, in their leasing counterparts resources are assumed to be rented only for certain durations.
In contrast to our model, in the leasing framework resources can not be rented for arbitrary durations.
Instead, there is a given number $K$ of different leases with individual costs and durations for which a resource can be leased.
The leasing model was introduced by Meyerson \cite{meyerson05} and problems like \textsc{FacilityLeasing} or \textsc{SetCoverLeasing} have been studied in \cite{anthony07,kling12,abshoff14} afterward.

A last problem worthmentioning here is the problem of \emph{scheduling with calibrations}\cite{calibration1,calibration2}.
Although it does not consider the aspect of minimizing resources and the number of machines is fixed, it is closely related to machine minimization and shares aspects with our model.
There is given a set of $m$ (identical) machines and a set of jobs with release times, deadlines and sizes.
After a machine is calibrated at a time $t$, it is able to process workload in the interval $[t, t+T]$, for some fixed $T$, and the goal is to minimize the number of calibrations.
For sufficiently large $m$, the problem is similar to a problem we need to solve in Section~\ref{sec:tentativeSchedules}.
\subsection{Our Results}
We study algorithms for \PROBLEM where jobs need to be scheduled on machines rented from the cloud such that the rental costs are minimized subject to quality of service constraints.
Particularly, the problem introduces the possibility for a scheduler to choose between different machine types being heterogeneous in terms of prices and processing capabilities.
It also captures the fact that due to times for preparing machines and acquiring resources, available computing power does not scale instantaneously.

Our results show the competitiveness to heavily depend on the minimum slack $\beta$ jobs exhibit.
While for $\beta < \setup{B}$ no finite competitiveness is possible,
a trivial rule achieves an optimal competitiveness of $\Theta(c\setup{A} + \setup{B})$ for $\beta = (1+\varepsilon)\setup{B}$, $0 \leq \varepsilon < \nicefrac{1}{\setup{B}}$.
    For $\nicefrac{1}{\setup{B}} \leq \varepsilon \leq 1$, we present an algorithm which is up to an $O(\nicefrac{1}{\varepsilon^2})$-factor optimal.
    Its competitiveness is $O(\nicefrac{c}{\varepsilon}  + \nicefrac{1}{\varepsilon^3})$, while a general lower bound is proven to be $\Omega(\nicefrac{c}{\varepsilon})$.

\section{Preliminaries}
\label{sec:pre}
We begin our study by showing that setup times are hard to cope with for any online algorithm when the minimum slack is below a natural threshold. 
This is the case because an online algorithm needs to hold machines ready at any time for arriving jobs with small deadlines.

\begin{proposition}
\label{ob:infeasible}
If $\beta < \setup{B}$, no online algorithm has a finite competitiveness.
\end{proposition}

\begin{proof}
Suppose there is an online algorithm with competitiveness $k$, where $k$ is some constant.
We observe that there needs to be a time $t$ at which the algorithm has no machine of type B open.
This is true because otherwise, the adversary can release only one job with 
$r_1 = 0$, $\size{1}{B} = 1$, $d_1 = \setup{B} + \size{1}{B} + \beta$ and $\size{1}{A} > \setup{B}-\setup{A}+\size{1}{B}+\beta$.
Then, the cost for the optimal algorithm is only $c \cdot (\setup{B}+1)$, whereas the online algorithm has cost of at least $c \cdot t$, which is a contradiction to the competitiveness as it grows in $t$ and hence, without bound.

Therefore, let $t$ be a time where the online algorithm has no machine of type B open.
The adversary releases one job $j$ with $r_j=t$, $d_j = r_j + \size{j}{B} + \beta $ and an arbitrary size $\size{j}{B}$ and $\size{j}{A} > \size{j}{B} + \beta$.
To finish this job without violating its deadline, it needs to be started not after $t+\beta$ on a machine of type B.
However, this is not possible since there is no machine of type B available at time $t$ and the setup of a type B machine needs time $\setup{B}>\beta$.
Hence, any online algorithm must rent a machine of type B all the time, which is a contradiction.
 \end{proof}

Due to this impossibility result, we restrict our studies to cases where the minimum slack is $\beta = (1+\varepsilon)\setup{B}$ for some $\varepsilon \geq 0$.
It will turn out that for very small values of $\varepsilon$, $0\leq \varepsilon < \nicefrac{1}{\setup{B}}$, there is a high lower bound on the competitiveness and essentially, we cannot do better than processing each job on its own machine.
However, for larger values of $\varepsilon$, the situation clearly improves and leaves room for designing non-trivial online algorithms.

\begin{lemma}
\label{le:lowerBound}
If $\beta = (1+\varepsilon)\setup{B}$, $\nicefrac{1}{\setup{B}} \leq \varepsilon \leq 1$, there is a lower bound on the competitiveness of $\Omega(\nicefrac{c}{\varepsilon})$.
For $0 \leq \varepsilon < \nicefrac{1}{\setup{B}}$ a lower bound is $\Omega(c\setup{A} + \setup{B})$.
\end{lemma}

\begin{proof}
We start with the case $\nicefrac{1}{\setup{B}} \leq \varepsilon \leq 1$.
Let $\setup{B} = \setup{A}$.
At $t=0$ the adversary starts the instance by releasing the first job with $\size{1}{A} = 1$, $\size{1}{B} = 1$ and $d_1 = \setup{B} +1 + \beta$.
Afterward, no job is released until time $t$ at which the online algorithm does not have any open machine.
Note that $t$ must exist as otherwise the online algorithm cannot be competitive at all.
At time $t$, the adversary releases $c\cdot t \cdot \lceil k \rceil$ jobs, where $k \coloneqq \frac{(1+2\varepsilon)\setup{B}}{2\varepsilon \setup{B} + \delta}$ for some sufficiently small $\delta > 0$, with the following properties.
All jobs have processing times $p_B = \varepsilon \setup{B}$ and $p_A = 2\varepsilon \setup{B} + \delta$. 
The deadline of all jobs is $d = t + p_B + \beta$ and hence, all jobs need to be finished by $t+(1+2\varepsilon)\setup{B}$.
Thus, an online algorithm cannot process these jobs on machines of type A and on a machine of type B opened at time $t$, it can process at most two jobs.
Hence, it has to rent $c \cdot t \cdot \lceil k \rceil/2$ machines of type B incurring cost of $\Omega(ctk \cdot c(\setup{B} + 2\varepsilon \setup{B}))$.
An offline algorithm can process $\lfloor k \rfloor$ jobs on a machine of type A and hence, all jobs on $O(ct)$ machines with cost $O(ct(\setup{A}+(1+2\varepsilon)\setup{B}))$.
This yields a competitiveness of
$
\Omega\left(\frac{kc(1+2\varepsilon) \setup{B}}{\setup{A}+(1+2\varepsilon)\setup{B}}\right) = \Omega(ck) = \Omega(c\varepsilon^{-1}).
$

Next, we consider the case $0 \leq \varepsilon < \nicefrac{1}{\setup{B}}$ and start with the bound of $\Omega(c\setup{A})$.
Let $\setup{B} = \setup{A}$.
At $t=0$ the adversary starts the instance by releasing the first job with $\size{1}{A} = 1$, $\size{1}{B} = 1$ and $d_1 = \setup{B} +1 + \beta$.
Afterward, the adversary does not release any job until time $t$ at which the online algorithm does not have any open machine.
Note that $t$ must exist as otherwise the online algorithm cannot be competitive at all.
At time $t$, the adversary releases $c\cdot t \cdot \lceil k \rceil$ jobs, where $k \coloneqq \frac{\setup{B} + 1}{1+\varepsilon \setup{B} + \delta}$, with the following properties.
All jobs have processing time $p_B = 1$ and $p_A = 1+\varepsilon \setup{B} +\delta$ for some sufficiently small $\delta >0$.
Furthermore, the deadline of all jobs is $d = t+1 + \beta$.
An online algorithm cannot process these jobs on machines of type A since setting up a machine takes too long to still be able to meet the deadlines.
Hence, it has to rent $c \cdot t \cdot \lceil k \rceil$ machines of type B incurring costs of $\Omega(ct \cdot ck(\setup{B} + 1))$.
An optimal offline algorithm can process all jobs on $O(ct)$ machines of type A with cost $O(ct\setup{B})$.
This yields a competitiveness of
$
\Omega(ck) = \Omega(c\setup{A}).
$

Next, we argue for the second bound of $\Omega(\setup{B})$, which is proven quite similar to the previous one. 
Let $\setup{A}$ and $\setup{B}$ be arbitrary.
At $t=0$ the adversary starts the instance by releasing the first job with $\size{1}{A} = 1$, $\size{1}{B} = 1$ and $d_1 = \setup{B} +1 + \beta$.
Afterward, the adversary does not release any job until time $t$ at which the online algorithm does not have any open machine.
At time $t$, the adversary releases $t$ jobs with the following properties.
All jobs have processing time $p_B = 1$ and $p_A> \setup{B} + p_B+1$.
Furthermore, the deadline of all jobs is $d = t +p_B + \beta$.
While the online algorithm needs one machine of type B per job, an optimal offline algorithm can schedule them using only $O(t/\setup{B})$ machines, each processing $\Theta(\setup{B})$ jobs.
Hence, the competitiveness is lower bounded by $\Omega\left(\frac{ct(\setup{B} +1)}{ct/\setup{B} \cdot \setup{B}}\right) = \Omega(\setup{B})$.

\end{proof}

\subsection{Simple Heuristics}
\label{subsec:greedy}
As a first step of studying algorithms in our model we discuss some natural heuristics.
One of the doubtlessly most naive rules simply decides on the machine type to process a job on based on the cost it incurs on this type.
The algorithm $A_1$ assigns each job to its own machine and chooses this machine to be of type A if $\setup{A} + \size{j}{A} \leq c(\setup{B} + \size{j}{B})$ and $\slack{j}{A} \geq \setup{A}$, or if $\slack{j}{B}<\setup{B}$. 
Otherwise, it chooses the machine to be of type B.

\begin{proposition}
\label{prop:naiveAlg}
If $\beta = (1+\varepsilon)\setup{B}$ and $0\leq \varepsilon < \nicefrac{1}{\setup{B}}$, $A_1$ is $\Theta(c\setup{A} + \setup{B})$-competitive.
For $\nicefrac{1}{\setup{B}} \leq \varepsilon \leq 1$ its competitiveness is $\Theta(\nicefrac{c}{\varepsilon} + \setup{B})$.
\end{proposition}

\begin{proof}
For $\tau, \tau' \in \{A,B\}$, let $J_{\tau,\tau'} \subseteq J$ be the set of jobs assigned to machines of type $\tau$ by $A_1$ and to machines of type $\tau'$ by \OPT.
Then, for the competitiveness of $A_1$ we obtain an upper bound of
\begin{align*}
\frac{\sum_{j \in J_{A,A}}(\size{j}{A} + \setup{A})}{\setup{A} + \sum_{j \in J_{A,A}}\size{j}{A}} &+ \frac{\sum_{j \in J_{A,B}}(\size{j}{A} + \setup{A})}{c(\setup{B} + \sum_{j \in J_{A,B}} \size{j}{B})} \\ 
&+\frac{c(\sum_{j \in J_{B,B}}(\size{j}{B}+\setup{B}))}{c(\setup{B} + \sum_{j \in J_{B,B}} \size{j}{B})}+  \frac{c(\sum_{j \in J_{B,A}}(\size{j}{B}+\setup{B}))}{\setup{A} + \sum_{j \in J_{B,A}}\size{j}{A}} \enspace .
\end{align*}
Observe that if a job $j$ is not assigned to the more cost-efficient machine type $\tau$, it holds $\size{j}{\tau} \geq \size{j}{\tau'}+\beta-\setup{\tau}$, implying $\setup{\tau'} \leq \setup{\tau}-\size{j}{\tau'} + \size{j}{\tau}$ as well.
This directly follows from the fact that $\slack{j}{\tau} < \setup{\tau}$ and $\slack{j}{\tau'} \geq \beta$.
Hence, the first three summands are bounded by $O(\setup{B})$.
Furthermore, we can bound the last summand by $s_B + \nicefrac{c}{\varepsilon}$ when $\varepsilon \setup{B} \geq 1$ and by $c\setup{A}$ otherwise.

This bound can easily be seen to be tight.
Consider an instance where all jobs need to be processed on machines of type B with sizes $\size{j}{B} = 1 \enspace \forall j \in J$ and let them arrive such that all can be assigned to one machine.
Then the above term $\frac{c(\sum_{j \in J_{B,B}}(\size{j}{B} +\setup{B}))}{c(\setup{B} + \sum_{j \in J_{B,B}} \size{j}{B})}$ is lower bounded by $\frac{1}{2} \setup{B}$ for $n\geq \setup{B}$.
Together with the lower bound from Lemma~\ref{le:lowerBound}, this proves the proposition.
 \end{proof}

While this trivial rule is optimal for $0 \leq \varepsilon < \nicefrac{1}{\setup{B}}$, for larger values of $\varepsilon$ the dependence of the competitiveness on the setup time is undesired as it can be quite high and in particular, it is sensitive to the time scale (and recall that we chose a time scale such that the smallest processing time of any job is at least $1$).
Therefore, the rest of the paper is devoted to finding an algorithm with a competitiveness being independent of $\setup{B}$ and narrowing the gap to the lower bound of $\Omega(\nicefrac{c}{\varepsilon})$ for $\nicefrac{1}{\setup{B}} \leq \varepsilon \leq 1$.

One shortcoming of $A_1$ is the fact that jobs are never processed together on a common machine.
A simple idea to fix this and to batch jobs is to extend $A_1$ by an \textsc{AnyFit} rule as known from bin packing problems (e.g.\ see \cite{sgall14}).
Such an algorithm dispatches the jobs one by one and only opens a new machine if the job to be assigned cannot be processed on any already open machine. 
Then, the job is assigned to some machine it fits into.
It turns out that the competitiveness of this approach still depends on the setup time $\setup{B}$, which we show by proving a slightly more general statement.
Consider the class \textsc{GreedyFit} consisting of all deterministic heuristics $\textsc{Alg}$ fitting into the following framework:

\begin{enumerate}
  \item At each time $t$ at which the jobs $J_t = \{j \in J : r_j = t\}$ arrive, $\textsc{Alg}$ processes them in \emph{any} order $j_1, j_2, \ldots$
  \begin{enumerate}[label={\arabic{enumi}.\arabic*}] 
    \item If $j_i$ cannot \emph{reasonably} be processed on any open machine, $\textsc{Alg}$ opens \emph{some} machine.
    \item $\textsc{Alg}$ assigns $j_i$ to \emph{some} open machine.
  \end{enumerate}
  \item $\textsc{Alg}$ \emph{may} close a machine as soon as it is about to idle.
\end{enumerate}
The terms \emph{some}, \emph{any} and \emph{may} above should be understood as to be defined by the concrete algorithm.
A job is said to be \emph{reasonably} processable on a machine $M$ if it does not violate any deadline and the processing cost it incurs does not exceed the cost for setting up and processing the job on a new machine.
Unfortunately, one can still construct bad instances leading to Proposition~\ref{th:greedyfit}.
\begin{proposition}
\label{th:greedyfit}
If $\beta \geq \setup{B}$, any \textsc{GreedyFit} algorithm is $\Omega(\nicefrac{c}{\varepsilon} + \setup{B}$)-competitive.
\end{proposition}

\begin{proof}
Fix an arbitrary $\setup{B}$ and set $\setup{A}=1$ and $c=1$.
At time $t=0$ release a job $j_0$ with $d_{j_0} = \setup{B} ^2$ and $\size{j_0}{A} = \setup{B}$.
Since the considered algorithm is deterministic and only $\size{j_0}{B}$ is not yet fixed, we can now determine and distinguish the following two cases: 
a) $\exists x \geq 1$ such that if $\size{j_0}{B} = x$, $\textsc{Alg}$ opens a machine of type $A$, and the alternative case
b) such an $x$ does not exist. 
We start with the former case
and set $\size{j_0}{B} = x$.
At time $t = \varepsilon$, $\setup{B} - 1$ many jobs $j$ with $\size{j}{A} = \setup{B}$, $\size{j}{B}  = 1$ and $d_j = \setup{B}^2$ are released.
On this instance, $\textsc{Alg}$ schedules all jobs on one machine of type A, incurring cost of $\setup{A} + \setup{B} ^2$.
Depending on whether $x<\setup{A} + \setup{B}$ or not, $\OPT$ schedules all jobs on one machine of type B or one machine of type A and one of type B.
In both cases it holds $c \cdot r_{max} \leq \OPT < 4\setup{B}$, proving this case.
In case b), $\textsc{Alg}$ always uses a machine of type B.
Hence, we can set $\size{j_0}{B}$ arbitrarily high and together with the lower bound from Lemma~\ref{le:lowerBound}, this concludes the proof.
 \end{proof}

\section{An $O(\nicefrac{c}{\varepsilon}+\nicefrac{1}{\varepsilon^3}$)-competitive Algorithm}
Due to the discussed observations, it seems that decisions on the type of a machine to open as well as finding a good assignment of jobs requires more information than given by a single job.
The rough outline of our approach is as follows. 
In a first step, we identify a variant of \PROBLEM that we can solve by providing an Integer Linear Program (ILP). 
This formulation is heavily based on some structural lemma that we prove next before turning to the actual ILP.
In a second step, we describe how to use the ILP solutions, giving infeasible subschedules, to come up with a competitive algorithm.
Throughout the description we make the assumption that $\varepsilon$ is known to the algorithm in advance, which, however, can easily be dropped by maintaining a guess on $\varepsilon$.

We now show a fundamental lemma that provides a way of suitably batching the processing of jobs and structuring the rental intervals.
We define intervals of the form $[is_\tau, (i+1)s_\tau)$, for $i \in \mathbb{N}_0$ and $\tau \in \{A,B\}$, as the \emph{$i$-th $\tau$-interval}.
Intuitively, we can partition the considered instance into subinstances each consisting only of jobs released during one $B$-interval and the times during which machines are open are aligned with these $\tau$-intervals.

\begin{lemma}
\label{le:intervals}
By losing a constant factor, we may assume that \OPT fulfills the following properties:
\begin{enumerate}
  \item each job $j$ that is processed on a machine of type $\tau$ and fulfills $p_{j,\tau} \geq s_\tau$ is assigned to an exclusive machine,
  \item each remaining machine $M$ of type $\tau$ is open for exactly five $\tau$-intervals, i.e., $[a_M, b_M) = [is_\tau, (i+5)s_\tau)$ for some $i \in \mathbb{N}_0$, and
  \item if it is opened at $a_M = (i-1)s_\tau$, the beginning of the $(i-1)$-th $\tau$-interval, it only processes jobs released during the $i$-th $\tau$-interval.
\end{enumerate}
\end{lemma}

\begin{proof}
We prove the lemma by describing how to modify \OPT to establish the three properties while increasing its cost only by a constant factor.
Consider a fixed type $\tau \in \{A,B\}$ and 
let $J_\tau \subseteq J$ be the set of jobs processed on machines of type $\tau$ in \OPT. 
Let $J_M$ be the set of jobs processed by a machine $M$.

\emph{Property 1. }
Any job $j \in J_\tau$ with $p_{j,\tau} \geq s_\tau$ is moved to a new exclusive machine of type $\tau$ if not yet scheduled on an exclusive machine in \OPT.
This increases the cost due to an additional setup and the resulting idle time by $s_\tau + p_{j,\tau} \leq 2 p_{j,\tau}$.
Applying this modification to all jobs $j$ with $p_{j,\tau} \geq s_\tau$ therefore increases the cost of \OPT by a factor of at most three and establishes the desired property.
From now on, we will assume for simplicity that all jobs $j$ scheduled on a machine of type $\tau$ fulfill $p_{j,\tau} < s_\tau$.

\emph{Property 2. }
Next, we establish the property that each remaining machine is open for exactly four $\tau$-intervals;
when establishing the third property, this is extended to five intervals as claimed in the lemma.
Consider a fixed machine $M$.
Partition the time during which $M$ is open into intervals of length $s_\tau$ by defining $I_k \coloneqq [a_M+ks_\tau, a_M+(k+1)s_\tau)$ for $k\in \{1,\ldots, \lceil(b_M-a_M)/s_\tau -1\rceil\}$.
Let $t_j$ be the starting time of job $j \in J_M$ and let $T_k = \{j \in J_M : t_j \in I_k\}$ partition the jobs of $J_M$ with respect to the interval during which they are started.
We replace machine $M$ by $\lceil(b_M-a_M)/s_\tau -1\rceil$ machines $M_1,M_2, \ldots$ such that $M_i$ processes the jobs from $T_i$ keeping the jobs' starting times as given by \OPT and setting $a_{M_i} \coloneqq \min_{j \in T_i} t_j-s_\tau$ and $b_{M_i} \coloneqq \max_{j \in T_i} (t_j + \size{j}{\tau})$.
Observe that the cost originally incurred by $M$ is $c_\tau(b_M - a_M)$ while those incurred by the replacing machines $M_1, M_2, \ldots$ are at most $c_\tau(b_M - a_M) + c_\tau(\lceil(b_M - a_M)/s_\tau -1\rceil)s_\tau \leq 2 c_\tau(b_M -a_M)$ where the additional term stems from the additional setups. 
Since it holds that $p_{j,\tau} < s_\tau$ for all $j \in J_M$, we conclude that $b_{M_i}-a_{M_i} \leq 3s_\tau$ for all $M_i$.
For each $M_i$, we can now simply decrease $a_{M_i}$ and increase $b_{M_i}$ such that $M_i$ is open for exactly four $\tau$-intervals.
This maintains the feasibility and increases the cost incurred by each machine $M_i$ from at least $c_\tau s_\tau$ to at most $4c_\tau s_\tau$.
Applying these modifications to all machines $M$, establishes the claimed property while increasing the overall cost by a factor of at most eight.

\emph{Property 3. }
It remains to prove the third property.
Again consider a fixed machine $M$.
Let $N_{i,M} \coloneqq \{j \in J_M : r_j \in [is_\tau, (i+1)s_\tau)\}$, $i \in \mathbb{N}_0$, be the (sub-)set of jobs released during the $i$-th $\tau$-interval and processed by $M$.
Furthermore, let $\#(M) \coloneqq |\{i : N_{i,M} \neq \emptyset\}|$ be the number of  $\tau$-intervals from which $M$ processes jobs.
If $\#(M) \leq 3$, we replace $M$ by $\#(M)$ machines $M_i$ each processing only jobs from $N_{i,M}$.
To define the schedule for $M_i$, let $N_{i,M} = \{j_{i_1}, j_{i_2}, \ldots\}$ such that $t_{i_1} < t_{i_2} < \ldots$ holds.
We reset the starting times of the jobs in $N_{i,M}$ by setting $t_{i_1} \coloneqq r_{i_1}$ and $t_{i_k} \coloneqq \max\{r_{i_k}, t_{i_{k-1}}+\size{i_{k-1}}{\tau}\}$, for $k>1$, and set $a_M \coloneqq (i-1)s_\tau$ and $b_M \coloneqq a_M+5s_\tau$.
This gives a feasible schedule for $N_{i,M}$ since no starting time is increased and according to Property~2 we have $\sum_{j \in N_{i,M}} p_{j,\tau} \leq 3s_\tau$ and hence, $\max_{j \in N_{i,M}} r_j + \sum_{j \in N_{i,M}} p_{j,\tau} < (i+4)s_\tau = b_M$.
Also, the replacing machines fulfill the three properties and the cost is increased by a factor less than four.

For the complementary case where $\#(M) > 3$, we argue as follows.
Observe that due to the already established property that $M$ is open for exactly four $\tau$-intervals, 
$N_{i,M} = \emptyset$ for all $i \geq a_M/s_\tau+4$.
Note that $a_M/s_\tau \in \mathbb{N}_0$ due to Property~2.
Also, for each $i \in \{a_M/s_\tau-1, a_M/s_\tau, a_M/s_\tau+2,a_M/s_\tau+3 \}$, we can move the jobs from $N_{i,M}$ to new machines by an argument analogous to the one given for the case $\#(M) \leq 3$ above.
Hence, we have established the property that $M$ only processes jobs $j \in  \bigcup_{i=0}^{a_M/s_\tau-2} N_{i,M} \cup N_{a_M/s_\tau+1,M}$ and applying the modification to all machines $M$ from the set $\mathcal{M}$ of machines fulfilling $\#(M) > 3$, increases the cost by a constant factor.

In order to finally establish the third property, our last step proves how we can reassign all those jobs $j$ with $j \in \bigcup_{i=0}^{a_M/s_\tau-2} N_{i,M}$ for some $M \in \mathcal{M}$.
Let $J' = \bigcup_{M \in \mathcal{M}}\bigcup_{i=0}^{a_M/s_\tau-2} N_{i,M}$ be the set containing these jobs and partition them according to their release times by defining $N_i = (\bigcup_{M \in \mathcal{M}} N_{i,M}) \cap J'$.
Note that for any $j \in N_i$ it holds $d_j \geq a_M + \setup{\tau} \geq (i+3)s_\tau$, for all $i \in \mathbb{N}_0$.
Let $w_i \coloneqq \sum_{j \in N_i} p_{j,\tau}$.
We can assign all jobs from $N_i$ to new machines fulfilling the three properties as follows:
We open $\lceil w_i/s_\tau \rceil$ new machines at time $(i-1)s_\tau$ and keep them open for exactly five $\tau$-intervals.
Due to the fact that for all jobs $j \in N_i$ it holds $r_j\leq (i+1)s_\tau$ and $d_j \geq (i+3)s_\tau$, we can accommodate a workload of at least $s_\tau$ in the interval $[(i+1)s_\tau, (i+3)s_\tau]$ on each machine by assigning jobs from $N_i$ to it in any order.
By these modifications the cost increase due to cases where $w_i \geq s_\tau$ is given by an additive of at most $c_\tau \sum_{i : w_i \geq s_\tau}\lceil w_i/s_\tau\rceil5s_\tau$, which is $O(\OPT)$ since $\OPT \geq c_\tau \sum_{i : w_i \geq s_\tau}w_i$.
The overall increase in the cost due to the cases where $w_i < s_\tau$ is given by an additive of $c_\tau \sum_{i : w_i <s_\tau} 5s_\tau = O(\OPT)$
due to our assumption $\OPT = \Omega(c \cdot r_{max})$.
 \end{proof}

By Lemma~\ref{le:intervals}, we can partition any instance into subinstances such that the $i$-th subinstance consists of those jobs released during the interval $[i \setup{B}, (i+1)\setup{B})$ and solve
them separately.
Hence, we assume without loss of generality that the entire instance only consists of 
jobs released during the interval $[0,\setup{B})$.

\subsection{Tentative Subschedules}
\label{sec:tentativeSchedules}
Before we turn to describing our algorithm for \PROBLEM, we first identify and describe a special problem variant that will be helpful when designing our algorithm.
In this variant, we assume setup costs instead of setup times, i.e., setups do not take any time but still incur the respective cost.
Under this relaxed assumption, the goal is to compute a schedule for a whole batch of (already arrived) jobs.
Precisely, for a given time $t$, we consider a subset $J' \subseteq J$ of jobs such that $r_j \leq t$ for all $j \in J'$ and our goal is to compute a schedule for $J'$ that minimizes the rental cost and in which each job is finished at least $\setup{B}$ time units before its deadline (has \emph{earliness} at least $\setup{B}$).
We do not require the starting times of jobs to be at least $t$ and hence, a resulting schedule may be infeasible for \PROBLEM if realized as computed.
Therefore, to emphasize its character of not being final, we call such a schedule a \emph{tentative schedule}.

We now divide the jobs into three sets according to their sizes.
Let $J_1 \coloneqq \{ j \in J : \size{j}{A} \geq \setup{A} \wedge \size{j}{B} \geq \setup{B}\}$ contain those jobs for which the processing cost dominates the setup cost on both machine types.
The set $J_2 \coloneqq \{ j \in J : \size{j}{A} < \setup{A} \wedge  \size{j}{B} < \setup{B}\}$ contains those jobs for which the processing cost does not dominate the setup cost on either of the two machine types. 
Finally, $J_3 \coloneqq J \setminus (J_1 \cup J_2)$ and we write $ J_3 = J_{3,1} \cup J_{3,2} = \{ j \in J_3 : \size{j}{A} \geq \setup{A} \wedge \size{j}{B} < \setup{B}\} \cup  \{ j \in J_3 :\size{j}{A} < \setup{A} \wedge \size{j}{B} \geq \setup{B}\}$.
By Lemma~\ref{le:intervals}, we directly have the following lemmas.

\begin{lemma}
\label{le:j2mm}
For an optimal schedule for jobs from $J_2$, we may assume each machine of type $\tau \in \{A,B\}$ to be open for exactly
five $\tau$-intervals.
If a job $j \in J_2$ is processed on a machine of type $\tau$ and $r_j \in [is_\tau,(i+1)s_\tau)$, $i \in \mathbb{N}_0$, its processing interval is completely contained in the interval $[is_\tau,(i+4)s_\tau]$.
\end{lemma}

\begin{lemma}
\label{le:j3mm}
For an optimal schedule for jobs from $J_{3,1}$, we may assume each machine of type $B$ to be open for exactly
five $B$-intervals.
If a job $j \in J_{3,1}$ is processed on a machine of type $A$, it is processed on an exclusive machine and if $j$ is processed on a machine of type $B$ and $r_j \in [i\setup{B},(i+1)\setup{B})$, $i \in \mathbb{N}_0$, its processing interval is completely contained in the interval $[i\setup{B},(i+4)\setup{B}]$.

Analogous statements hold for jobs from $J_{3,2}$.
\end{lemma}

By Lemma~\ref{le:j2mm} and Lemma~\ref{le:j3mm}, we can now formulate the problem (which is NP-hard by the NP-hardness of classical \textsc{BinPacking})  as an ILP (cf.\ Fig.~\ref{fig:ilp1}) and compute $O(1)$-approximate tentative schedules.
We use $\mathcal{I}_\tau(j)$ to denote all possible processing intervals of job $j$ on a machine of type $\tau \in \{A,B\}$.
Intuitively, $\mathcal{I}_A(j) \cup \mathcal{I}_B(j)$ describes all possible ways how job $j$ can be scheduled.
Note that $\mathcal{I}_A(j)$ and $\mathcal{I}_B(j)$ are built under the assumptions from Lemma~\ref{le:intervals}, Lemma~\ref{le:j2mm} and Lemma~\ref{le:j3mm}.
\begin{figure}
	\rule{\textwidth}{0.4pt}
\begin{align}
\text{min} \quad 5 \setup{A} \sum_i z_i + 5 c\setup{B} \cdot z_B +\sum_{\substack{x(I,j):  I \in \mathcal{I}_A(j) \\j \in J_1 \cup J_{3,1}}} 
 \hspace{-.45cm} & x(I,j)\size{j}{A}  + c 
\sum_{\substack{x(I,j): I \in \mathcal{I}_B(j) \\j \in J_1 \cup J_{3,2}}} x(I,j)\size{j}{B} \hspace{-1.5cm} \qquad \notag\\
\text{s.t.\ }
	\sum_{\substack{j \in J_2 \cup J_{3,2}: \\ r_j \in [(i-1)\setup{A},i\setup{A})}} \sum_{I \in \mathcal{I}_A(j): t \in I} x(I,j) &\leq z_i 
	\qquad\qquad\qquad\quad\ \forall t \in L_i, \forall i 
	\\
	\sum_{j \in J_2 \cup J_{3,1}} \sum_{I \in \mathcal{I}_B(j): t \in I} x(I,j) &\leq z_B \qquad\qquad\qquad\qquad \forall t \in L_B\\
	 \sum_{I \in \mathcal{I}_A(j) \cup \mathcal{I}_B(j)} x(I,j) &= 1 \qquad\qquad\qquad\qquad\quad\ \forall j \in J\\
	x(I,j) &\in \{0,1\}  \qquad\qquad\qquad \forall j \in J, \forall I 
	\end{align}
\rule{\textwidth}{0.4pt}
\caption{Integer Linear Program for variant with setup cost.}
\label{fig:ilp1}
\end{figure}
For each $I \in \mathcal{I}_A(j) \cup \mathcal{I}_B(j)$, the indicator variable $x(I,j)$ states if job $j$ is processed in interval $I$.
We use $L_B$ to denote all left endpoints of intervals in $\bigcup_{j \in J_2 \cup J_{3,1}} \mathcal{I}_B(j)$ and $L_i$ to denote all left endpoints of intervals in $\bigcup_{j \in J_2 \cup J_{3,2}: r_j\in[(i-1)\setup{A}, i\setup{A})} \mathcal{I}_A(j)$.
Additionally, we use a variable $z_B$ to denote the number of (non-exclusive) machines of type B that we rent. 
The variable $z_i$ describes the number of machines of type A that we rent and that process jobs released during the 
$(i-1)$-th A-interval.
For simplicity we assume that $\setup{B}$ is an integer multiple of $\setup{A}$. 

We are now asked to minimize the cost for machines of type B plus those for machines of type A, taking into account that each machine of type $\tau$ is either open for a duration of exactly $5s_\tau$ time units (first two summands of the objective function) or is an exclusive machine (last two summands).
The constraints of type~(1) and type~(2) ensure that, at any point in time, the number of jobs processed on (non-exclusive) machines does not exceed the number of open machines.
Additionally, constraints of type~(3) and type~(4) ensure that each job is completely processed by exactly one machine in a contiguous interval.

Observe that in general this ILP has an infinite number of variables since each $\mathcal{I}_\tau(j)$ contains all possible processing intervals of $j$ on a machine of type $\tau$.
Yet, there is an efficient way to reduce the number of variables that need to be considered to $O(|J|^2)$ such that afterward a solution only being by a constant factor larger than the optimal one of the original formulation exists.
\begin{lemma}
\label{le:polyintervals}
By losing a constant factor, we may assume $|\bigcup_{j \in J} (\mathcal{I}_A(j) \cup \mathcal{I}_B(j))| = O(|J|^2)$.
\end{lemma}

\begin{proof}
The following proof is an extension of one from \cite{chuzhoy04}, which studies a related issue for the problem of machine minimization.

We show how to reduce the number $|\bigcup_{j \in J} (\mathcal{I}_A(j) \cup \mathcal{I}_B(j)|$ of \emph{job intervals} to $O(|J|^2)$.
Note that due to Lemma~\ref{le:intervals}, Lemma~\ref{le:j2mm} and Lemma~\ref{le:j3mm}, it holds 
\begin{align*}
\left| \bigcup_{\tau \in \{A,B\}} \bigcup_{j \in J_1} \mathcal{I}_\tau(j)\right| = O(1),
\left| \bigcup_{j \in J_{3,1}} \mathcal{I}_A(j)\right| = O(1) \text{ and } \left| \bigcup_{j \in J_{3,2}} \mathcal{I}_B(j)\right| = O(1) \enspace .
\end{align*}
For the sake of simplicity, we only argue that $|\bigcup_{j \in J_2} (\mathcal{I}_A(j) \cup \mathcal{I}_B(j))| = O(|J|^2)$ in the following.
The reasoning for the remaining sets $J_{3,1}$ and $J_{3,2}$ is analogous but even simpler as it is a more restricted case than the one for $J_2$ due to the aforementioned bounds on the number of job intervals on one of the two machine types.

We now describe the construction and afterward prove the claimed bound on the possible loss.
Recall Lemma~\ref{le:intervals} and Lemma~\ref{le:j2mm}.
Let $M_0$ be a machine of type B such that it is open until $4\setup{B}$ (which corresponds to machines represented by $z_B$ in Fig.~\ref{fig:ilp1}) and for each $i \in \{1, \ldots, \frac{\setup{B}}{\setup{A}}\}$ let $M_i$ be a machine of type A which can process jobs during the interval $[(i-1)\setup{A}, (i+3)\setup{B})$ (which corresponds to machines represented by $z_i$ in Fig.~\ref{fig:ilp1}).
We first construct sets $D_i \subseteq J_2$, $i \in \{0, 1, \ldots, \frac{\setup{B}}{\setup{A}} \}$, such that $D_i$ can be scheduled on machine $M_i$.
To do so, for each machine $M_i$, we separately consider the set $J_2$ of all jobs.
At each time machine $M_i$ gets idle, we greedily assign the job that will finish earliest among all jobs that are not yet scheduled on $M_i$ and that can still meet their deadlines.
Let $P_i$ be the set which contains all right endpoints of processing intervals of jobs scheduled on $M_i$ together with all release times and deadlines of the remaining jobs. 
Note that $|P_0| = O(|J|)$ and $|P_i| = O(|J|)$, for all $i \in \{1, \ldots, \frac{\setup{B}}{\setup{A}} \}$.
Now we reduce the number of job intervals of any job $j$ as follows and afterward bound the increase in the cost.
Each left endpoint of a job interval in $\mathcal{I}_B(j)$ is rounded down to its nearest point in $P_0$ and each right endpoint is rounded up to its nearest point in $P_0$. 
Similarly, each left endpoint of a job interval in $\mathcal{I}_A(j)$ is rounded down and each right endpoint is rounded up to its nearest point in its corresponding $P_i$.
By this construction, the number of job intervals becomes $O(|J|^2)$.

Now, consider an optimal solution $S$ for the original instance.
If a job $j \in D_0$ is scheduled on a machine of type A in $S$, remove it from $D_0$ and similarly,
if a job $j \in D_i$, $i \in \{1, \ldots,\frac{\setup{B}}{\setup{A}} \}$, is scheduled on a machine of type B in $S$, remove it from $D_i$.
Also, for each job $j \in D_i$, remove this job from $S$.
Now, all processing intervals of jobs $j\in S$ must contain at least one point from the respective set $D_i$ by construction.
Thus, after performing the modifications of the release times and deadlines, the jobs $j\in S$ can always be scheduled by twice the number of type A and type B machines used in $S$.
By definition, the remaining jobs in the sets $D_i$ can be assigned to one machine for each set.
Note that due to the removal of jobs from $D_0$, this set is empty if the original solution $S$ had no machine of type B.
Similarly, for any fixed $i \in \{1, 2, \ldots, \frac{\setup{B}}{\setup{A}}\}$, $D_i$ is empty if the original solution $S$ had no machine of type A processing jobs released during the interval $[(i-1)\setup{A}, i\setup{A})$.
Hence, there is a schedule for the modified instance with at most three times the cost of the original schedule. 
\end{proof}

\subsection{The \ALGO Algorithm}
\label{sec:bars}
In this section, we describe and analyze our competitive algorithm, which is essentially based on several observations concerning how to restrict and then turn tentative schedules, as described in the previous section, into feasible solutions.
We first show the following lemma, which relates the cost of an optimal schedule to the cost of one where jobs are finished early.
For simplicity we use $\Delta \coloneqq \frac{1}{2} \varepsilon \setup{B}$.

\begin{lemma}
\label{le:premature}
There is a schedule with all jobs having earliness at least $\setup{B}$, cost $O((\nicefrac{c}{\varepsilon} + \nicefrac{1}{\varepsilon^2})\OPT)$ and no machine processes any two jobs $j, j'$ with $r_j \in [i \Delta, (i+1) \Delta)$ and $r_{j'} \in [i'\Delta, (i'+1)\Delta)$ with $i \neq i'$ and $i,i' \in \mathbb{N}_0$.

Furthermore, the cost for jobs $j \in J' \coloneqq \{ j : \exists \tau \in \{A,B\}~r_j+\size{j}{\tau} > d_j -\setup{B}\}$ is $O((\nicefrac{c}{\varepsilon} + \nicefrac{1}{\varepsilon^2})\OPT)$.
The cost for jobs $j \in J \setminus J'$ is $O((\nicefrac{1}{\varepsilon^2})\OPT)$.
\end{lemma}

\begin{proof}
Given an optimal schedule $\OPT$ for $J$, we show how to modify it such that the desired properties and bound on the cost hold.

First, let $E_\tau \subseteq J'$ be the set of jobs $j \in J'$ that are processed on a machine of type $\tau$ in \OPT and for which it holds $r_j+\size{j}{\tau} > d_j -\setup{B}$.
By increasing the cost by a factor of $O(\nicefrac{c}{\varepsilon})$ we can assume, for all $\tau \in \{A,B\}$, that each job $j \in E_\tau$ is processed on an exclusive machine of type $\tau' \neq \tau$.
This is true since $\slack{j}{\tau'} \geq \beta$ and hence $\size{j}{\tau} \geq  \size{j}{\tau'} + \varepsilon \setup{B}$.
Consequently, all jobs from $E_A \cup E_B$ fulfill the desired properties.

Therefore, consider the remaining set of jobs $J \setminus (E_A \cup E_B)$.
By assuming that each machine only processes jobs released during a common interval $[i\Delta, (i+1)\Delta)$, $i \in \mathbb{N}_0$, the cost are increased by a factor of $O(\nicefrac{1}{\varepsilon})$. 
Consider any such machine $M^i$.
Let $J_k^{M^i}$ be the set of jobs finished in the interval $[k\varepsilon \setup{B}, (k+1)\varepsilon \setup{B})$ on machine $M^i$, $k \in \mathbb{N}_0$.
By losing a factor of $O(\nicefrac{1}{\varepsilon})$ we can replace $M^i$ by machines $M^i_k$ such that $M^i_k$ processes exactly those jobs from $J_k^{M^i}$ that are started and finished in $[k\varepsilon \setup{B}, (k+1)\varepsilon \setup{B})$.
Now we can manipulate the starting times of jobs such that we obtain a schedule as desired in the lemma.
To this end, we perform the following steps:
Let ${k_1}, {k_2}, \ldots$ be the jobs processed on machine $M^i_k$ such that for the starting times $t_j$ it holds $t_{k_1} < t_{k_2} < \ldots$
In case $k\varepsilon\setup{B} \geq \setup{B} + (i+1)\Delta$, 
let $d = t_{k_1} - (i+1)\Delta$ and reset the starting time $t_j$ of job $j$ on $M^i_k$ to $t_j = t_j - d$.
Observe that $t_j \geq r_j$ for all $j \in \{k_1,k_2,\ldots\}$ and each job $j$ is finished by $t_j-d+p_j \leq d_j - d \leq d_j - k\varepsilon \setup{B} + (i+1)\Delta \leq d_j -\setup{B}$.
In case $(i+1)\Delta \leq k\varepsilon \setup{B} < \setup{B} + (i+1)\Delta$, we assume $r_{k_1} \leq r_{k_2} \leq \ldots$, which is without loss of generality.
Let $d = t_{k_1} - r_{k_1}$.
We reset the starting time $t_j$ of job $j$ on $M^i_k$ to $t_j = t_j - d$.
Observe that $t_j \geq r_j$ for all $j \in \{k_1,k_2,\ldots\}$ and each job is finished at least $\setup{B}$ before its deadline (due to the ordering of jobs and the fact that the workload of all of them is not larger than $\varepsilon \setup{B}$).
In case $k\varepsilon \setup{B} < (i+1)\Delta$, it follows $k\varepsilon \setup{B} \leq i\Delta$ and we do not need to do anything.
Hence, we have established a schedule as desired by the lemma.
 \end{proof}

We are now ready to describe our approach.
The formal description is given in Figure~\ref{fig:algorithm}.
It relies on $\Delta$ and its relation to the slack of a job (after decreasing the $d_j$ in Step~\ref{step:deadlines}) and uses the following partition of $J$ into the sets: 
\begin{align*}
&J_1 = \{j : \slack{j}{B} \geq 2\Delta \wedge (\slack{j}{A} \geq 2\Delta \vee \size{j}{A} \leq \Delta)\},\\
 &J_2 = \{j  : \slack{j}{B} \geq 2\Delta \wedge (\slack{j}{A} < 2\Delta \wedge \size{j}{A} > \Delta) \},\\
&J_3 = \{j  : \slack{j}{B} < 2\Delta  \wedge \slack{j}{A} \geq 2\Delta\}.
\end{align*}
Before proving the correctness and bounds on the cost, we shortly describe the high level ideas skipping technical details, which will become clear during the analysis.
The algorithm proceeds in phases, where each phase is devoted to scheduling jobs released during an interval of length $\Delta$.
In a given phase, we first decrease the deadlines of jobs by $\setup{B}$ (cf.\ Step~\ref{step:1}) to ensure that later on (cf.\ Step~\ref{step:realize}) tentative schedules meeting these modified deadlines can be extended by the necessary setups without violating any original deadline.
We also precautionary open machines for jobs that are required to be started early.
Then at the end of the phase (cf.\ Step~\ref{step:2}), we compute tentative schedules with additional restrictions on starting times and machines to use (cf.\ Step~\ref{step:restrictions}), 
which are carefully defined depending on the characteristics of jobs concerning their slacks.
This approach ensures that we can turn solutions into feasible schedules while guaranteeing that costs are not increased too much.
The feasibility is crucial since tentative schedules are not only delayed due to the added setups but also because of computing the schedules only at the end of a phase.

In the next lemmas, the analysis of the algorithm is carried out, leading to our main result in Theorem~\ref{th:theorem}.
Recall that due to Lemma~\ref{le:premature} it is sufficient to show that the schedules for a single phase are feasible and the costs are increased by a factor of $O(1)$ and $O(\nicefrac{1}{\varepsilon})$, respectively, in comparison to the tentative schedules if they were computed without the additional restrictions of Step~\ref{step:restrictions}.
\begin{figure}[h!]
\rule{\textwidth}{0.4pt}
\textsc{BatchedDispatch}($\varepsilon$)\\
Let $C = (\setup{B} \geq \setup{A} + \Delta \wedge \frac{\setup{A}}{c} > \Delta)$.\\
In phase $i\in \mathbb{N}$ process all jobs $j$ with $r_j\in [(i-1)\Delta, i\Delta)$:
  \begin{enumerate}
    \item Upon arrival of a job $j$ at time $t$ \label{step:1}
    \begin{enumerate}[label={\arabic{enumi}.\arabic*}] 
      \item Set $d_j = d_j - \setup{B}$. \label{step:deadlines}
      \item Classify $j$ to belong to the set $J_1, J_2$ or $J_3$.\\
      If C holds, also define $J_{2,1} \coloneqq J_2$ and $J_{2,2} \coloneqq \emptyset$.\\
      Otherwise, $J_{2,1} \coloneqq \{j \in J_2 : c\size{j}{B} > \setup{A}\}$ and $J_{2,2} \coloneqq J_2 \setminus J_{2,1}$.
      \item If $C$ does not hold, open a machine of type A at $t$ for each $j \in J_{2,1}$. \label{openstep}
    \end{enumerate}
    \item At time $i \Delta$ \label{step:2}
    \begin{enumerate}[label={\arabic{enumi}.\arabic*}] 
      \item Compute tentative schedules for the job sets with these additional restrictions:\label{step:restrictions} 
      \begin{itemize}
	\item Starting times of jobs from $J_1, J_{2,2}$ and $J_3$ are at least $i\Delta$,
	\item starting times of jobs from $J_{2,1}$ on machines of type B are at least $i\Delta$ and
	\item jobs from $J_{2,2}$ ($J_3$) only use machines of type B (type A).
      \end{itemize}
          \item Realize the tentative schedules by increasing the starting times of jobs from
          \begin{itemize}
	    \item $J_1$ by $s_\tau$ on machines of type $\tau$,
	    \item $J_{2,1}$ by $\setup{B}$ on machines of type B and by $\setup{A}+\Delta$ or $\setup{A}$ depending on whether $C$ holds or not on machines of type A,
	    \item $J_{2,2}$ by $\setup{B}$ and
	    \item $J_3$ by $\setup{A}$,
          \end{itemize}
          and doing the respective setups at time $i\Delta$.
	Use machines from Step~\ref{openstep} for $J_{2,1}$ if necessary and otherwise close them after finishing the setups. \label{step:realize}
    \end{enumerate}
  \end{enumerate}
  \vspace{-.4cm}
\rule{\textwidth}{0.4pt}
\caption{Description of \textsc{BatchedDispatch}($\varepsilon$) algorithm for \PROBLEM.}
\label{fig:algorithm}
\end{figure}

\begin{lemma}
\label{le:j1}
For jobs from $J_1$, \ALGO produces a feasible schedule with rental cost $O(\nicefrac{c}{\varepsilon} + \nicefrac{1}{\varepsilon^2})\OPT$.
\end{lemma}
\begin{proof}
Fix an arbitrary phase $i$ and define $N_i \coloneqq \{j \in J_1 : r_j \in [(i-1)\Delta, i \Delta)\}$ to be the set of jobs released during the $i$-th phase.
Let $S$ be the tentative schedule for $N_i$ and let $t_j$ denote the starting time of job $j \in N_i$ in $S$.
Let $S'$ and $t'_j$ be defined analogously for the respective schedule produced by \ALGO.

We first reason about the feasibility of $S'$.
First of all, the tentative schedule $S$ provides a solution with $t_j \geq i\Delta$ and $t_j + \size{j}{\tau} \leq d_j$ for all $j \in N_1$ processed on a machine of type $\tau$, i.e., no modified deadline (as given after Step~\ref{step:1} is executed) is violated and no job is started before time $i\Delta$.
Since the starting times are increased in $S'$ to $t'_j = t_j + \setup{\tau}$, each job $j$ processed on a machine of type $\tau$ is not started before $i\Delta + \setup{\tau}$ and is finished by $t'_j + \size{j}{\tau} \leq t_j + \setup{\tau} + \size{j}{\tau} \leq d_j + \setup{\tau}$.
Therefore, by starting the setups at $i\Delta$, $S'$ is a feasible solution.

Second, we have to prove the bound on the cost of $S'$.
Let $S^*$ be the tentative schedule for $N_i$ if computed without the restrictions on the starting times.
Since the cost of $S'$ are not larger than those of $S$, we need to show that the cost of $S$ only increase by a constant factor in comparison to $S^*$, which according to Lemma~\ref{le:premature} has cost $O((\nicefrac{c}{\varepsilon} + \nicefrac{1}{\varepsilon^2})\OPT)$.
Let $N^\tau_i \subseteq N_i$ be the set of jobs processed on machines of type $\tau$ in $S^*$, for $\tau \in \{A,B\}$.
Consider an arbitrary machine $M$ of type $\tau$ in $S^*$ and let $J_M = \{j_1, j_2, \ldots\}$ denote the jobs processed on $M$ such that $t^*_{j_1} < t^*_{j_2} < \ldots$ holds.
Let $k \in \mathbb{N}$ such that $t^*_{j_{k}}+p_{j_k,\tau}<i\Delta$ and $t^*_{j_{k+1}} +p_{j_{k+1},\tau}\geq i\Delta$.
Now, we can leave all jobs $j \in \{j_{k+2}, j_{k+3}, \ldots \}$ unaffected (i.e.\ $t_j = t^*_j$) and we can process $j_{k+1}$ on an exclusive machine started at
time $t_{j_{k+1}} =  t^*_{j_{k+1}} = i\Delta > r_{j_{k+1}} $.
This is feasible since $j_{k+1}$ is finished at time $i\Delta + p_{j_{k+1},\tau} \leq d_{j_{k+1}}$.
The last inequality either holds since $p_{j_{k+1}} \leq \Delta$ and $d_{j_{k+1}} \geq 2\Delta + r_{j_{k+1}}$ because $\slack{j_{k+1}}{B} \geq 2\Delta$, or since $\sigma_{j_{k+1}, \tau} = d_{j_{k+1}} - r_{j_{k+1}} - p_{j_{k+1},\tau} \geq 2 \Delta$.
All remaining jobs $j \in \{j_1, \ldots, j_k\}$ can be moved to a new machine $M'$ of type $\tau$ and 
$t_j = t^*_j + \Delta$.
Each job $j \in \{j_1, \ldots, j_k\}$ is then finished at time $t_j + p_{j,\tau} = t^*_j + p_{j,\tau} + \Delta \leq d_j$.
The last inequality either holds since $d_{j} \geq 2\Delta + r_{j} \geq (i+1) \Delta \geq t^*_j + p_{j,\tau} + \Delta$, or since $t^*_j < r_j + \Delta$ and $d_j - r_j - p_{j,\tau} \geq 2 \Delta$.
Also, $t^*_{j_q} +  \Delta + p_{j_q,\tau} \leq t^*_{j_{q+1}} + \Delta$ for all $q \in \{1, \ldots, k-1\}$.
Hence, $t_j \geq i\Delta$ for all $j \in J_M$ and the cost at most triples in comparison to $S^*$ since the two additional machines need not run longer than $M$.
Applying the argument to all machines $M$ gives $S \leq 3S^* =O((\nicefrac{c}{\varepsilon} + \nicefrac{1}{\varepsilon^2})\OPT)$, which concludes the proof.
 \end{proof}

\begin{lemma}
\label{le:j2}
For jobs from $J_2$ and $J_3$, \ALGO produces a feasible schedule with rental cost $O(\nicefrac{c}{\varepsilon} + \nicefrac{1}{\varepsilon^3})\OPT$.
\end{lemma}

\begin{proof}
We first consider the set $J_2$ and argue about jobs from $J_{2,1}$ and $J_{2,2}$ separately. 

\emph{Scheduling jobs from $J_{2,1}$. }
Fix an arbitrary phase $i$ and let $N_i \coloneqq \{j \in J_{2,1} : r_j \in [(i-1)\Delta, i \Delta)\}$.
Let $S$ be the tentative schedule for $N_i$ and let $t_j$ denote the starting time of job $j \in N_i$ in $S$.
Let $S'$ and $t'_j$ be defined analogously for the respective schedule produced by \ALGO.
Let $N^\tau_i \subseteq N_i$ be the set of jobs processed on machines of type $\tau$ in $S$, for $\tau \in \{A,B\}$.
Recall that we compute a tentative schedule for jobs from $J_{2,1}$ at time $i\Delta$ with the additional constraint $t_j \geq i\Delta$ for all $i \in N_i^B$, that is, the starting times of all jobs processed on machines of type B are at least $i\Delta$.
Similar to the previous lemma, with respect to machines of type B, $S'$ is feasible and fulfills the bound on the cost.
For machines of type A we can argue as follows.
In case $C$ holds, we have $t'_j = t_j + \Delta + \setup{A} \geq i\Delta + \setup{A}$ for all $j \in N_i^A$.
Hence, we are able to realize the respective schedule by starting the respective setup processes at time $i\Delta$. 
Also, since $\Delta+\setup{A} \leq \setup{B}$ by the fact that $C$ holds, $t'_j + \size{j}{A} \leq t_j + \Delta + \setup{A} + \size{j}{A} \leq d_j + \Delta + \setup{A} \leq d_j + \setup{B}$, i.e., no original deadline is violated.
Hence, $S'$ is feasible and the bound on the cost is satisfied by a reasoning as in Lemma~\ref{le:j1}.

In case $C$ does not hold, recall that at each arrival of a job $j \in J_{2,1}$, we open a new machine of type A at time $r_j$.
This ensures that a machine is definitely available for the respective job $j \in N_i^A$ at time $t_j + \setup{A}$.
Therefore, we obtain a feasible schedule and it only remains to prove the bound on the cost.
Note that the additional setup cost in $S'$ compared to $S$ is upper bounded by $|J_{2,1}| \cdot \setup{A}$.
Because $\setup{A} < c\size{j}{B}$ holds and $S$ can process at most three jobs on a machine of type A (since $\slack{j}{A} = d_j-r_j-\size{j}{A} < 2\Delta$ implying $d_j < 2\Delta + \size{j}{A} +r_j$ and $\size{j}{A} > \Delta$ by the definition of $J_2$), $|J_{2,1}| \cdot \setup{A} \leq 3S$ holds.
Therefore, \ALGO produces a feasible schedule for jobs from $J_{2,1}$ incurring cost of $O(\nicefrac{c}{\varepsilon} + \nicefrac{1}{\varepsilon^2})\OPT$.

\emph{Scheduling jobs from $J_{2,2}$. }
Fix an arbitrary phase $i$ and let $N_i \coloneqq \{j \in J_{2,2} : r_j \in [(i-1)\Delta, i \Delta)\}$.
For these jobs we claim that the tentative schedule $S$ has cost $O(\nicefrac{c}{\varepsilon} + \nicefrac{1}{\varepsilon^3})\OPT$.
Then, the same reasoning as in Lemma~\ref{le:j1} proves the desired result.
Hence, we only have to show the claim.
To this end, we show how to shift any workload processed on a machine of type A in a tentative schedule $S^*$ corresponding to $S$ without the additional restriction on machine types to use to machines of type B.
Recall that we only need to consider the case that $C$ does not hold since otherwise $J_{2,2} = \emptyset$.
Let $N_i^A \subseteq N_i$ be the jobs processed on machines of type A in $S^*$.
Note that $N_i^A \cap \{ j : \exists \tau \in \{A,B\}~r_j+\size{j}{\tau} > d_j -\setup{B}\} = \emptyset$ and thus according to Lemma~\ref{le:premature}, the cost for jobs from $N_i^A$ is $O(\nicefrac{1}{\varepsilon^2})\OPT$.
We show that only using machines of type B increases the cost by $O(\nicefrac{1}{\varepsilon})$ and distinguish two cases depending on whether $\frac{\setup{A}}{c} \leq \Delta$ holds.

First, we show the claim if $\frac{\setup{A}}{c} \leq \Delta$ holds.
On the one hand, any machine of type A in $S^*$ can process at most three jobs $j$ with $j \in N_i^A$ (since $\slack{j}{A} = d_j-r_j-\size{j}{A} < 2\Delta$ implying $d_j < 2\Delta + \size{j}{A} +r_j$ and $\size{j}{A} > \Delta$ by the definition of $J_2$), leading to (setup) cost of $\Omega(|N_i^A|\setup{A})$.
On the other hand, we can schedule all jobs from $N_i^A$ on $O(\frac{|N_i|\setup{A}}{\Delta c})$ machines of type B each open for $O(\setup{B})$ time units and thus, with cost of $O(\frac{|N_i^A|\setup{A}}{\Delta c}c\setup{B})$.
To see why this is true, observe that for all jobs $j \in N_i^A$ it holds $[r_j,d_j] \supseteq [i\Delta, (i+1)\Delta) \eqqcolon I_i$ since $r_j \leq i\Delta$ and $\slack{j}{A}\geq 2\Delta$.
Because all jobs $j \in N_i^A$ fulfill $\size{j}{B} \leq \setup{A}/c$ by definition, we can accommodate $\lfloor c\Delta/\setup{A}\rfloor \geq 1$ many jobs from $N_i^A$ in $I_i$.
Hence, there is a schedule with cost $O(\nicefrac{1}{\varepsilon^3})\OPT$ only using machines of type B, proving the claim.

Finally, it needs to be proven that the claim also holds for the case $\frac{\setup{A}}{c} > \Delta$ and hence, $\setup{B} < \setup{A} + \Delta$.
Let $M_1, \ldots, M_m$ denote the machines of type A used by $S^*$ and let $\kappa_1,\ldots, \kappa_m$ be the durations for which they are open.
Since $\size{j}{B} \leq \size{j}{A}$ for all $j\in J_{2,2}$, we can replace each machine $M \in \{M_1, \ldots, M_m\}$ by a machine $M'$ of type B using the same schedule on $M'$ as on $M$.
We increase the cost by a factor of at most $\frac{\sum_{i=1}^m c(\setup{B}+ \kappa_i)}{\sum_{i=1}^m (\setup{A} + \kappa_i)} \leq \frac{c \setup{B}}{\setup{A}} + c \leq \frac{c (\Delta +\setup{A})}{\setup{A}} + c \leq \frac{c\Delta}{\setup{A}} + 2c= O(\setup{A}/\Delta) = O(\nicefrac{1}{\varepsilon})$, where the second last bound holds due to $\setup{A}/c > \Delta$.
Thus, we have shown a schedule fulfilling the desired properties and cost $O(\nicefrac{1}{\varepsilon^3})\OPT$ to exist and since $S$ fulfills the same bound on the cost, this proves the claim.

Since all jobs from $J_3$ can be assumed to be scheduled on machines of type A without any loss in the cost, we obtain the same result for $J_3$ as we have in Lemma~\ref{le:j1}.
 \end{proof}

\begin{theorem}
\label{th:theorem}
For $\beta = (1+\varepsilon)\setup{B}$, $\nicefrac{1}{\setup{B}} \leq \varepsilon \leq 1$, \ALGO is $O(\nicefrac{c}{\varepsilon} + \nicefrac{1}{\varepsilon^3})$-competitive.
\end{theorem}

\section{Conclusion}
We presented a competitive algorithm for \PROBLEM where jobs with deadlines need to be scheduled on machines rented from the cloud so as to minimize the rental cost.
We parameterized instances by their minimum slack $\beta$ and showed different results depending thereon.
Alternatively, we could examine the problem without this instance parameter and instead understand $\beta$ as a parameter of the algorithm describing the desired maximum tardiness of jobs.

Although the considered setting with $k=2$ types of machines seems to be restrictive, it turned out to be challenging and closing the gap between our lower and upper bound remains open.
Also, it is in line with other research (e.g.\ \cite{azar13,saha13}) and in this regard it is a first step toward models for scheduling machines from the cloud addressing the heterogeneity of machines.

For future work, however, it would be interesting to study even more general models for arbitrary $k>1$.

\end{document}